\documentclass{article}

\usepackage{times}
\usepackage{soul}
\usepackage{url}
\usepackage[hidelinks]{hyperref}
\usepackage[utf8]{inputenc}
\usepackage[small]{caption}
\usepackage{graphicx}
\usepackage{amsmath}
\usepackage{amsthm}
\usepackage{booktabs}
\usepackage{algorithm}
\usepackage{algorithmic}
\usepackage{authblk}

\newtheorem{theorem}{Theorem}

\title{Behavioral QLTL
}

\author{
Giuseppe De Giacomo, Giuseppe Perelli
}

\affil{
	Sapienza University of Rome \\ 
	\{degiacomo, perelli\}@diag.uniroma1.it
}

\date{}

\usepackage{paralist}
\usepackage[inline]{enumitem}

\usepackage{xargs}

\usepackage{xspace}

\usepackage{xstring}

\usepackage{boolexpr}

\usepackage[cmex10]{mathtools}
\usepackage{amsfonts}
\usepackage{mathrsfs}
\usepackage{latexsym}
\usepackage{textcomp}
\usepackage{pifont}
\usepackage{amssymb,bm,xargs,amsthm}
\usepackage{enumitem}

\usepackage{lipsum}
\usepackage{todonotes}
\usepackage{cancel}

\newtheorem{definition}{Definition}
\newtheorem{lemma}{Lemma}

\newtheorem{corollary}{Corollary}

\newcommand{\ignore}[1]{}

\newcommand{\argemp}[2]{\if&#1&\else#2\fi}

\newcommand{\argdef}[2]{\if&#1&#2\else#1\fi}

\newcommand{\argint}[3]{\if&#2&\else#1#2#3\fi}

\newcommand{\argext}[3]{\if&#1&#3\else#1\if&#3&\else#2#3\fi\fi}

\newcommandx{\mthfnt}[3][1=, 2=0]{{
	\IfStrEqCase{#1}
	{%
		{}%
		{#3}%
		{Name}%
		{%
			\IfStrEqCase{#2}
			{%
				{0}{\mathcal{#3}}%
				{1}{\mathscr{#3}}%
				{2}{\mathfrak{#3}}%
				{3}{\mathbb{#3}}%
			}
			[\ensuremath{\clubsuit}]%
		}%
		{Set}%
		{%
			\IfStrEqCase{#2}
			{%
				{0}{\mathrm{#3}}%
				{1}{\mathsf{#3}}%
				{2}{\mathbb{#3}}%
				{3}{\mathbf{#3}}%
			}
			[\ensuremath{\clubsuit}]%
		}%
		{Fun}%
		{%
			\IfStrEqCase{#2}
			{%
				{0}{\mathsf{#3}}%
				{1}{\mathrm{#3}}%
			}
			[\ensuremath{\clubsuit}]%
		}%
		{Rel}%
		{%
			\IfStrEqCase{#2}
			{%
				{0}{\mathit{#3}}%
				{1}{\mathtt{#3}}%
			}
			[\ensuremath{\clubsuit}]%
		}%
		{Sym}%
		{%
			\IfStrEqCase{#2}
			{%
				{0}{\mathtt{#3}}%
				{1}{\mathbf{#3}}%
			}
			[\ensuremath{\clubsuit}]%
		}%
		{Elm}%
		{\mathnormal{#3}}
	}
[\ensuremath{\clubsuit}]%
}}

\newcommand{\mthsub}[1]{\argemp{#1}{\ensuremath{_{\mathnormal{#1}}}}}

\newcommand{\mthsup}[1]{\argemp{#1}{\ensuremath{^{\mathnormal{#1}}}}}

\newcommandx{\mth}[5][1=, 2=0, 4=, 5=]{{\ensuremath{\mthfnt[#1][#2]{#3}\mthsub{#4}\mthsup{#5}}}}

\newcommandx{\mtharg}[6][1=, 2=0, 4=, 5=]{{\mth[#1][#2]{#3}[#4][#5]\ensuremath{\argint{(}{#6}{)}}}}

\newcommand{\mthempty}{\mth[][]}

\newcommand{\mthstyname}{0}
\newcommand{\mthname}[1][]{\mth[Name][\argdef{#1}{\mthstyname}]}

\newcommand{\mthstyset}{0}
\newcommand{\mthset}[1][]{\mth[Set][\argdef{#1}{\mthstyset}]}

\newcommand{\mthstyfun}{0}
\newcommand{\mthfun}[1][]{\mth[Fun][\argdef{#1}{\mthstyfun}]}

\newcommand{\mthstysym}{0}
\newcommand{\mthsym}[1][]{\mth[Sym][\argdef{#1}{\mthstysym}]}

\newcommand{\tuple}[1]
{\ensuremath{\!\argint{\langle}{#1}{\rangle}}}

			\newcommandx{\AName}[4][1=, 2=, 3=, 4=]{\mthname[#4]{A#3}[#1][#2]}
			\newcommandx{\BName}[4][1=, 2=, 3=, 4=]{\mthname[#4]{B#3}[#1][#2]}
			\newcommandx{\CName}[4][1=, 2=, 3=, 4=]{\mthname[#4]{C#3}[#1][#2]}
			\newcommandx{\DName}[4][1=, 2=, 3=, 4=]{\mthname[#4]{D#3}[#1][#2]}
			\newcommandx{\EName}[4][1=, 2=, 3=, 4=]{\mthname[#4]{E#3}[#1][#2]}
			\newcommandx{\FName}[4][1=, 2=, 3=, 4=]{\mthname[#4]{F#3}[#1][#2]}
			\newcommandx{\GName}[4][1=, 2=, 3=, 4=]{\mthname[#4]{G#3}[#1][#2]}
			\newcommandx{\HName}[4][1=, 2=, 3=, 4=]{\mthname[#4]{H#3}[#1][#2]}
			\newcommandx{\IName}[4][1=, 2=, 3=, 4=]{\mthname[#4]{I#3}[#1][#2]}
			\newcommandx{\JName}[4][1=, 2=, 3=, 4=]{\mthname[#4]{J#3}[#1][#2]}
			\newcommandx{\KName}[4][1=, 2=, 3=, 4=]{\mthname[#4]{K#3}[#1][#2]}
			\newcommandx{\LName}[4][1=, 2=, 3=, 4=]{\mthname[#4]{L#3}[#1][#2]}
			\newcommandx{\MName}[4][1=, 2=, 3=, 4=]{\mthname[#4]{M#3}[#1][#2]}
			\newcommandx{\NName}[4][1=, 2=, 3=, 4=]{\mthname[#4]{N#3}[#1][#2]}
			\newcommandx{\OName}[4][1=, 2=, 3=, 4=]{\mthname[#4]{O#3}[#1][#2]}
			\newcommandx{\PName}[4][1=, 2=, 3=, 4=]{\mthname[#4]{P#3}[#1][#2]}
			\newcommandx{\QName}[4][1=, 2=, 3=, 4=]{\mthname[#4]{Q#3}[#1][#2]}
			\newcommandx{\RName}[4][1=, 2=, 3=, 4=]{\mthname[#4]{R#3}[#1][#2]}
			\newcommandx{\SName}[4][1=, 2=, 3=, 4=]{\mthname[#4]{S#3}[#1][#2]}
			\newcommandx{\TName}[4][1=, 2=, 3=, 4=]{\mthname[#4]{T#3}[#1][#2]}
			\newcommandx{\UName}[4][1=, 2=, 3=, 4=]{\mthname[#4]{U#3}[#1][#2]}
			\newcommandx{\VName}[4][1=, 2=, 3=, 4=]{\mthname[#4]{V#3}[#1][#2]}
			\newcommandx{\WName}[4][1=, 2=, 3=, 4=]{\mthname[#4]{W#3}[#1][#2]}
			\newcommandx{\XName}[4][1=, 2=, 3=, 4=]{\mthname[#4]{X#3}[#1][#2]}
			\newcommandx{\YName}[4][1=, 2=, 3=, 4=]{\mthname[#4]{Y#3}[#1][#2]}
			\newcommandx{\ZName}[4][1=, 2=, 3=, 4=]{\mthname[#4]{Z#3}[#1][#2]}
\newcommand{\defeq}{\stackrel{\triangle}{=}}
\renewcommand{\defeq}{\doteq}

\newcommand{\set}[2]
{\ensuremath{\argint{\{}{\argext{#1}{\allowbreak \mid \allowbreak}{#2}}{\}}}}

\newcommand{\card}[1]
{\mthempty{\argint{\vert}{#1}{\vert}}}

\newcommand{\rst}{\upharpoonright}

\newcommand{\SetN}
{\mthset[2]{N}}

\newcommand{\LTL}{\textsc{ltl}\xspace}

\newcommand{\FOL}{\textsc{fol}\xspace}
\newcommand{\MSO}{\textsc{mso}\xspace}

\newcommand{\QLTL}{{\textsc{qltl}}\xspace}
\newcommand{\QPTL}{{\textsc{qptl}}\xspace}
\newcommand{\BQLTL}{\textsc{qltl$_{\mthfun{B}}$}\xspace}
\newcommand{\WBQLTL}{\textsc{qltl$_{\mthfun{WB}}$}\xspace}
\newcommand{\SL}{{\textsc{sl}}\xspace}
\newcommand{\ATLS}{{\textsc{atl$^{\star}$}}\xspace}

\newcommand{\Tree}{\TName}
\newcommand{\dir}{\mthfun{dir}}
\newcommand{\hide}{\mthfun{hide}}
\newcommand{\xray}{\mthfun{xray}}

\newcommand{\shape}{\mthfun{shape}}
\newcommand{\change}{\mthfun{change}}
\newcommand{\child}{\mthfun{child}}
\newcommand{\ndet}{\mthfun{ndet}}

\newcommand{\DPW}{\textsc{dpw}\xspace}

\newcommand{\X}{\mthsym{X}}

\newcommand{\F}{\mthsym{F}}
\newcommand{\G}{\mthsym{G}}
\newcommand{\U}{\mthsym{U}}
\newcommand{\R}{\mthsym{R}}

\newcommand{\true}{\mthsym{true}\xspace}
\newcommand{\false}{\mthsym{false}\xspace}

\newcommand{\cmodels}{\models_{\mthsym{C}}}

\newcommand{\smodels}{\models_{\mthsym{S}}}

\newcommand{\bmodels}{\models_{\mthsym{B}}}

\newcommand{\wbmodels}{\models_{\mthsym{WB}}}

\newcommand{\free}{\mthfun{free}}

\newcommand{\qntPref}{\mthsym{\wp}}

\newcommand{\qnt}{\mthsym{Qn}}

\newcommand{\Automaton}{\AName}

\newcommand{\NAutomaton}{\NName}

\newcommand{\DAutomaton}{\DName}

\newcommand{\Bool}{\mathcal{B}}

\newcommand{\Lang}{\LName}

\newcommand{\strategy}{\mthfun{s}}

\newcommand{\Var}{\mthsym{Var}}

\newcommand{\Dep}{\mthfun{Dep}}

\renewcommand{\Game}{\GName}
\newcommand{\Players}{\mthsym{Pl}}
\newcommand{\Actions}{\mthsym{Ac}}
\newcommand{\States}{\mthsym{St}}
\newcommand{\trnFun}{\mthfun{tr}}

\newcommand{\singleton}{\mthsym{singleton}}

\definecolor{almond}{rgb}{0.94, 0.87, 0.8}
\definecolor{beige}{rgb}{0.96, 0.96, 0.86}

\usepackage{tikz}
\usetikzlibrary{arrows,shapes}

\usepackage{wrapfig}

\tikzstyle{every node} =
[draw = none, fill = white, thin]
\tikzstyle{every edge} +=
[black, thick]

\tikzstyle{noall} =
[draw = none, fill = none]
\tikzstyle{nodraw} =
[draw = none, fill = white]
\tikzstyle{nofill} =
[draw = black, fill = none]

\tikzstyle{cnode} =
[circle, draw = black]
\tikzstyle{snode} =
[regular polygon, regular polygon sides = 4, draw = gray!50]
\tikzstyle{lnode} =
[diamond, draw = gray!75]
\tikzstyle{pnode} =
[regular polygon, regular polygon sides = 5, draw = gray]

\begin{document}

\maketitle

\begin{abstract}
  In this paper we introduce  Behavioral QLTL, which is a ``behavioral'' variant of
  linear-time temporal logic on infinite traces with second-order
  quantifiers. Behavioral QLTL is characterized by the fact
  that the functions that assign the truth value of the quantified
  propositions along the trace can only depend on the past. In other
  words such functions must be``processes''. This gives to the logic a
  strategic flavor that we usually associate to planning. Indeed we
  show that temporally extended planning in nondeterministic domains,
  as well as LTL synthesis, are expressed in Behavioral QLTL through
  formulas with a simple quantification alternation. While, as this
  alternation increases, we get to forms of planning/synthesis in
  which conditional and conformant planning aspects get mixed.  We
  study this logic from the computational point of view and compare it
  to the original QLTL (with non-behavioral semantics) and with
  simpler forms of behavioral semantics.
\end{abstract}

\section{Introduction}
\label{sec:introduction}

Since the very early time of AI, researchers have tried to reduce planning to logical reasoning, i.e., satisfiability, validity, logical implication \cite{Green69}.
However as we consider more and more sophisticated forms of planning this becomes more and more challenging, because the logical reasoning we need to do is intrinsically second-order.
One prominent case is if we want to express the model of the world (aka the environment) and the goal of the agent directly in Linear-time Temporal Logic, which is the logic used most in formal method to specify dynamic systems.
Examples are the pioneering work on using temporal logic as a sort of programming language through the MetateM framework~\cite{BFGGO95}, the work on temporal extended goals and declarative control constraints~\cite{BK98,BK00}, the work on planning via model-checking~\cite{CGGT97,CR00,DTV99,BCR01}, the work on adopting \LTL logical reasoning (plus some meta-theoretic manipulation) for certain forms of planning~\cite{CM98,CDV02}.
More recently the connection between planning in nondeterministic domains and (reactive) synthesis \cite{PR89} has been investigated, and in fact it has been shown that planning in nondeterministic domains can be seen in general terms as a form of synthesis in presence of a model of the environment~\cite{CBM19,AminofGMR19}, also related to synthesis under assumptions~\cite{CH07,ChatterjeeHJ08}. 

However the connection between planning and synthesis also clarifies formally that we cannot use directly the standard forms of reasoning in \LTL, such as satisfiability, validity, or  logical implication, to do planning.
Indeed the logical reasoning task we have to adopt is a nonstandard one, called “\emph{realizability}” \cite{Chu63, PR89}, which is in inherently a second-order form of reasoning on \LTL specifications.
So one question comes natural: can we use the second-order version of \LTL, called \QLTL (or \QPTL) \cite{Sis85} and then avoid use nonstandard form of reasoning?

In~\cite{CDV02} a positive answer was given limited to conformant planning, in which we cannot observe response of the environment to the agent actions. Indeed it was shown that conformant planning could be captured through standard logical reasoning in \QLTL.  But the results there do not extend to conditional planning (with or without full observability) in nondeterministic environment models.
The reason for this is very profound.
Any plan must be a ``\emph{process}’’, i.e., observe what has happened so far (the history), observe the current state and take a decision on the next action to do \cite{ALW89}. \QLTL instead interprets quantified propositions (i.e.,  in the case of planning, the actions to be chosen) through functions that have access to the whole traces, i.e., also the future instants, hence they cannot be considered processes.
This is a clear mismatch that makes standard \QLTL unsuitable to capture planning through standard reasoning tasks.

This mismatch is not only a characteristic of \QLTL, but, interestingly, even of logics that have been introduced specifically for in strategic reasoning. This has lead to investigating the ``\emph{behavioral}'' semantics in these logics.
In their seminal work~\cite{MMPV14}, Mogavero et al. introduce and analyze the behavioral aspects of quantification in Strategy Logic (\SL): a logic for reasoning about the strategic behavior of agents in a context where the properties of executions are expressed in \LTL.
They show that restricting to behavioral quantification of strategies is a way of both making the semantics more realistic and computationally easier.
In addition, they proved that behavioral and non-behavioral semantics coincide for certain fragments, including the well known \ATLS~\cite{AHK02}, but diverge for more interesting classes of formulas, e.g., the ones that can express game-theoretic properties such as Nash Equilibria and the like.
This has started a new line of research that aims at identifying new notions of behavioral and non-behavioral quantification, as well as characterize the syntactic fragments that are invariant to these semantic variations~\cite{GBM18,GBM20}.

In this paper we introduce a behavioral semantics for \QLTL. The resulting logic, called \emph{Behavioral-\QLTL} (\BQLTL)
is characterized by the fact that the functions that assign the truth value of the quantified propositions along the trace can only depend on the past.
In other words such functions must be ``\emph{processes}''.
This makes \BQLTL perfectly suitable to capture extended forms of planning through standard reasoning tasks (satisfiability in particular).

Indeed, temporally extended planning in nondeterministic domains, as well as \LTL synthesis, are expressed in \BQLTL through formulas with a simple quantification alternation.
While, as this alternation increases, we get to forms of planning/synthesis in which conditional and conformant planning aspects get mixed.
For example, the \BQLTL formula of the form $\exists Y \forall X \psi$ represents the conformant planning over the \LTL specification (of both environment model and goal) $\psi$, as it is intended in~\cite{Rin04} (note that this could be done also with standard \QLTL, since $\exists Y$ is put upfront as it cannot depend on the nondeterministic evolution of the fluents in the planning domain).  Instead, the \BQLTL formula $\forall X \exists Y \psi$ represents contingent planning, i.e., \emph{Planning in Fully Observable Nondeterministic Domains} (FOND), as well as \LTL synthesis (which, instead, could not be captured in standard \QLTL).  By taking \BQLTL formulas with increased alternation, one can describe more complex forms of planning and synthesis.  The \BQLTL formula $\forall X_1 \exists Y \forall X_2 \varphi$ represents the problem of \emph{Planning in Partially Observable Nondeterministic Domains} (POND), where $X_1$ and $X_2$ are the visible and hidden parts of the domain, respectively.  By going even further in alternation, we get a generalized form of POND where a number of actuators with hierarchically reduced visibility are coordinated to execute a plan that fulfills a temporally extended goal in an environment model. Interestingly this instantiates problems of distributed synthesis with hierarchical information studied in formal methods~\cite{PR90,KV01,FS05}.

 

%
%
%
%
%
%

We study \BQLTL, by introducing a formal semantics that is \emph{Skolem-based}, meaning that we make use of different notions of Skolem functions and Skolemization to define the truth-value of formulas.
The advantage of this approach is in the correspondence between Skolem functions and strategies/plans in synthesis and planning problems.
As a matter of fact, they can all be represented as suitable labeled trees, describing all the possible executions of a given process that receive inputs from the environment.
We show characterize the complexity of satisfiability in \BQLTL 
 is $(n+1)$-EXPTIME-complete, with $n$ being the number of quantification blocks of the form $\forall X_i\exists Y_i$ the in the formula.
This improves the complexity of the satisfiability problem for classic \QLTL, which depends on the overall quantifier alternation in the formula, and in particular is $2(n-1)$-EXSPACE-complete.
Moreover, it also shows that the corresponding synthesis and planning problems can be optimally solved in \BQLTL, as the matching lower-bound is provided by a reduction of these problems.

We also consider a weak variant of \BQLTL, called \emph{Weak Behavioral-\QLTL} (\WBQLTL), where the history is always visible while we have restriction visibility on the curent instant only.  We show that the complexity of satisfiability in \WBQLTL is $2$-EXPTIME-complete, regardless of the number and alternation of quantifiers.
The reason for this is in that processes are modeled in a way that they have full visibility on the past computation.
This allows them to find the right plan by means of a local reasoning, and so without employing computationally expensive automata projections.
As for the case of \BQLTL, such procedure is optimal to solve the corresponding synthesis problems, as the matching lower-bound is again provided by a reduction of them.

\section{Quantified Linear-Time Temporal Logic}
\label{sec:QLTL}

We introduce \emph{Quantified Linear-Temporal Logic} as an extension of \emph{Linear-Time Temporal Logic}.

\paragraph{Linear-Time Temporal Logic}
Linear Temporal Logic (\LTL) over infinite traces was originally
proposed in Computer Science as a specification language for
concurrent programs~\cite{Pnu77}.
Formulas of \LTL are built from a set $\Var$ of \emph{propositional variables} (or simply variables), together with Boolean and temporal operators.
Its syntax can be described as follows:

\begin{center}
	$\varphi ::= x \mid \neg \varphi \mid \varphi \vee \varphi \mid \varphi \wedge \varphi \mid \X \varphi \mid \varphi \U \varphi$
\end{center}
where $x \in \Var$ is a propositional variable.

Intuitively, the formula $\X \varphi$ says that $\varphi$ holds at the \emph{next} instant.
Moreover, the formula $\varphi_1 \U \varphi_2$ says that at some future instant $\varphi_2$ holds and \emph{until} that point, $\varphi_1$ holds.

We also use the standard Boolean abbreviations $\true:= x \vee \neg x$ (\emph{true}), $\false:= \neg \true$ (\emph{false}), and $\varphi_1 \to \varphi_2 := \neg \varphi_{1} \vee \varphi_{2}$ (\emph{implies}).
In addition, we also use the binary operator $\varphi_{1} \R \varphi_{2} \defeq \neg (\neg \varphi_{1} \U \neg \varphi_{2})$ (\emph{release}) and the unary operators $\F \varphi := \true \U \varphi$ (\emph{eventually}) and $\G \varphi := \neg \F \neg \varphi$ (\emph{globally}).

The classic semantics of \LTL is given in terms of infinite traces, i.e., truth-values over the natural numbers.
More precisely, an \emph{interpretation} $\pi: \SetN \to 2^{\Var}$ is a function that maps each natural number $i$ to a truth assignment $\pi(i) \in 2^{\Var}$ over the set of variables $\Var$.
Along the paper, we might refer to finite \emph{segments} of a computation $\pi$.
More precisely, for two indexes $i, j \in \SetN$, by $\pi(i,j) \defeq \pi(i), \ldots, \pi(j) \in (2^{\Var})^{*}$ we denote the finite segment of $\pi$ from it's $i$-th to its $j$-th position.
A segment $\pi(0,j)$ starting from $0$ is also called a \emph{prefix} and is sometimes denoted $\pi_{\leq j}$.

We say that an \LTL formula $\varphi$ is true on an assignment $\pi$ at instant $i$, written $\pi, i \cmodels \varphi$, if:

\begin{itemize}[label={-},leftmargin=*]
	\setlength\itemsep{0em}
	\item
		$\pi, i \cmodels x$, for $x \in \Var$ iff $x \in \pi(i)$;
	
	\item
		$\pi, i \cmodels \neg \varphi$ iff  $\pi, i \not \cmodels \varphi$;
		
	\item
		$\pi, i \cmodels \varphi_1 \vee \varphi_2$ iff either $\pi, i \cmodels \varphi_1$ or $\pi, i \cmodels \varphi_2$;
	
	\item
		$\pi, i \cmodels \varphi_1 \wedge \varphi_2$ iff both $\pi, i \cmodels \varphi_1$ and $\pi, i \cmodels \varphi_2$;
	
	\item
		$\pi, i \cmodels \X \varphi$ iff $\pi, i + 1 \cmodels \varphi$;
	
	\item
		$\pi, i \cmodels \varphi_1 \U \varphi_2$ iff for some $j \geq i$, we have that $\pi, j \cmodels \varphi_2$ and for all $k \in \{i, \ldots j - 1\}$, we have that $\pi, k \cmodels \varphi_1$.
		
\end{itemize}

A formula $\varphi$ is \emph{true} over $\pi$, written $\pi \cmodels \varphi$ iff $\pi, 0 \cmodels \varphi$.
A formula $\varphi$ is \emph{satisfiable} if it is true on some interpretation and \emph{valid} if it is true in every interpretation.


\paragraph{Quantified Linear-Time Temporal Logic}

Quantified Linear-Temporal Logic (\QLTL) is an extension of \LTL with two \emph{Second-order} quantifiers~\cite{SVW87}.
Its formulas are built using the classic \LTL Boolean and temporal operators, on top of which existential and universal quantification over variables is applied.
Formally, the syntax is given as follows:

\begin{center}
$\varphi ::= \exists x \varphi \mid \forall x \varphi \mid x \mid \neg \varphi \mid \varphi \vee \varphi \mid \varphi \wedge \varphi \mid \X \varphi \mid \varphi \U \varphi \mid \varphi \R \varphi$,
\end{center}
where $x \in \Var$ is a propositional variable.

Note that this is a proper extension of \LTL, as \QLTL has the same expressive power of \MSO~\cite{SVW87}, whereas \LTL is equivalent to \FOL~\cite{GPSS80}.

In order to define the semantics of \QLTL, we introduce some notation.
For an interpretation $\pi$ and a set of variables $X \subseteq \Var$, by $\pi_{\rst X}$ we denote the \emph{projection} interpretation over $X$ defined as $\pi_{\rst X}(i) \defeq \pi(i) \cap X$  at any time point $i \in \SetN$.
Moreover, by $\pi_{\rst -X} \defeq \pi_{\rst \Var \setminus X}$ we denote the projection interpretation over the complement of $X$.
For a single variable $x$, we simplify the notation as $\pi_{\rst x} \defeq \pi_{\rst \{ x \}}$ and $\pi_{\rst -x} \defeq \pi_{\rst \Var \setminus \{ x \}}$.
Finally, we say that $\pi$ and $\pi'$ \emph{agree} over $X$ if $\pi_{\rst X} = \pi'_{\rst X}$.

Observe that we can reverse the projection operation by combining interpretations over disjoint sets of variables.
More formally, for two disjoint sets $X, X' \subseteq \Var$ and two interpretations $\pi_{X}$ and $\pi_{X'}$ over $X$ and $X'$, respectively, $\pi_{X} \Cup \pi_{X'}$ is defined as the (unique) interpretation over $X \cup X'$ such that its projections on $X$ and $X'$ correspond to $\pi_{X}$ and $\pi_{X'}$, respectively.

The \emph{classic} semantics of the quantifiers in a \QLTL formula $\varphi$ over an interpretation $\pi$, at instant $i$, denoted $\pi, i \cmodels \varphi$, is defined as follows:

\begin{itemize}[label={-},leftmargin=*]
	\setlength\itemsep{0em}
	\item
		$\pi, i \cmodels \exists x \varphi$ iff there exists an interpretation $\pi'$ such that $\pi_{\rst -x} = \pi'_{\rst -x}$ and $\pi', i \cmodels \varphi$;
		
	\item
		$\pi, i \cmodels \forall x \varphi$ iff for every interpretation $\pi'$ such that $\pi_{\rst -x} = \pi'_{\rst -x}$, it holds that $\pi', i \cmodels \varphi$;
\end{itemize}

A variable $x$ is \emph{free} in $\varphi$ if it occurs at least once out of the scope of either $\exists x$ or $\forall x$ in $\varphi$.
By $\free(\varphi)$ we denote the set of free variables in $\varphi$.

As for \LTL, we say that $\varphi$ is true on $\pi$, and write $\pi \cmodels \varphi$ iff $\pi, 0 \cmodels \varphi$.
Analogously, a formula $\varphi$ is \emph{satisfiable} if it is true on some interpretation $\pi$, whereas it is \emph{valid} if it is true on every possible interpretation $\pi$.
Note that, as quantifications in the formula replace the interpretation over the variables in their scope, we can assume that $\pi$ are interpretations over the set $\free(\varphi)$ of free variables in $\varphi$.

A \QLTL formula is in \emph{prenex normal form} if it is of the form $\qntPref \psi$, where $\qntPref = \qnt_{1} x_{1} \ldots \qnt_{n} x_{n}$ is a \emph{prefix quantification} with $\qnt_{i} \in \{\exists, \forall \}$ and $x_i$ being a variable occurring on a \emph{quantifier-free} subformula $\psi$, which can be regarded as \LTL.
Every \QLTL formula can be rewritten in prenex normal form, meaning that it is true on the same set of interpretations.
Consider for instance the formula $\G (\exists y (y \wedge \X \neg y))$.
This is equivalent to $\forall x \exists y (\singleton(x) \to (\G(x \to (y \wedge \X \neg y))))$, with $\singleton(x) \defeq \F x \wedge \G(x \to \X \G \neg x)$ expressing the fact that $x$ is true exactly once on the trace~\footnote{The reader might observe that pushing the quantification over $y$ outside the temporal operator does not work. Indeed, the formula $\exists y \G(y \wedge \X \neg y)$ is unsatisfiable.}.
A full proof of the reduction to prenex normal form can be derived from~\cite[Section 2.3]{Tho97}.
For convenience and without loss of generality, from now on we will assume that \QLTL formulas are always in prenex normal form.
Recall that for a formula $\varphi = \qntPref \psi$ is easy to obtain the prefix normal form of its negation $\neg \varphi$ as $\overline{\qntPref} \neg \psi$, where $\overline{\qntPref}$ is obtained from $\qntPref$ by swapping every quantification from existential to universal and vice-versa.
From now on, by $\neg \varphi$ we denote its prenex normal form transformation.

An \emph{alternation} in a quantification prefix $\qntPref$ is either a sequence $\exists x \forall y$ or a sequence $\forall x \exists y$ occurring in $\qntPref$.
A formula of the form $\qntPref \psi$, is of \emph{alternation-depth} $k$ if $\qntPref$ contains exactly $k$ alternations.
By $k$-\QLTL we denote the \QLTL fragment of formulas with alternation $k$.
Moreover, $\Sigma_{k}^{\QLTL}$ and $\Pi_{k}^{\QLTL}$ denote the fragments of $k$-\QLTL of formulas starting with an existential and a universal quantification, respectively.

It is convenient to make use of the syntactic shortcuts $\exists X \varphi \defeq \exists x_{1} \ldots \exists x_{k} \varphi$ and $\forall X \varphi \defeq \forall x_{1} \ldots \forall x_{k} \varphi$ with $X = \{x_1, \ldots, x_k\}$.
Formulas can then be written in the form $\qnt_{1} X_1 \ldots \qnt_{n} X_n \psi$ such that every two consecutive occurrences of quantifiers are in alternation, that is, $\qnt_{i} = \exists$ iff $\qnt_{i + 1} = \forall$, for every $i \leq n$.


The satisfiability problem consists into, given a \QLTL formula $\varphi$, determine whether it is satisfiable or not.
Note that every formula $\varphi$ is satisfiable if, and only if, $\exists \free(\varphi) \varphi$ is satisfiable.
This means that we can study the satisfiability problem in \QLTL for \emph{closed} formulas, i.e., formulas where every variable is quantified.

Such problem is decidable, though computationally highly intractable in general~\cite{SVW87}.
For a given natural number $k$, by $k$-EXPSPACE we denote the language of problems solved by a Turing machine with space bounded by $2^{2^{\ldots^{2^{n}}}}$, where the height of the tower is $k$ and $n$ is the size of the input.
By convention $0$-EXPSPACE denotes PSPACE.

\begin{theorem}[\cite{Sis85}]
	\label{thm:qltlsatisfiability}
	The satisfiability problem for $k$-\QLTL formulas is $k$-EXPSPACE-complete.
\end{theorem}

\section{Skolem Functions for QLTL Semantics}
\label{sec:skolemization}

We now give an alternative way to capture the semantics of \QLTL, which is in terms of (second order) Skolem functions.
This will allow us later to suitably restrict such Skolem function to capture behavioral semantics, by forcing them to depend only form the past history and the current situation.

Let $\qntPref$ be a quantification prefix.
By $\exists(\qntPref)$ and $\forall(\qntPref)$ we denote the set of variables that are quantified existentially and universally, respectively.
Moreover, by $X <_{\qntPref} Y$ we denote the fact that $X$ occurs \emph{before} $Y$ in $\qntPref$.
For a given set of consecutive variables $Y \in \exists(\qntPref)$ that are existentially quantified, by $\Dep_{\qntPref}(Y) = \set{X \in \forall(\qntPref)}{X <_{\qntPref} Y}$ we denote the set of variables to which $Y$ depends on in $\qntPref$.
Moreover, for a given set $F \subseteq \Var$ of variables by $\Dep_{\qntPref}^{F}(Y) = F \cup \Dep_{\qntPref}(Y)$ we denote the \emph{augmented dependency}, taking into account an additional set of variables for dependency.
Whenever clear from the context, we omit the subscript and simply write $\Dep(Y)$ and $\Dep^{F}(Y)$.

The relation defined above captures the concept of \emph{functional dependence} generated by quantifiers and free variables in a \QLTL formula.
Intuitively, whenever a dependence occurs between two variables $X$ and $Y$, this means that the existential choices in $Y$ are determined by a function whose domain is given by all possible choices available in $X$, be it universally quantified or free in the corresponding formula.
This dependence is know in first-order logic as \emph{Skolem function} and can be described in \QLTL as follows.

\begin{definition}[Skolem function]
	\label{def:skolem}
	For a given quantification prefix $\qntPref$ defined over a set $\Var(\qntPref) \subseteq \Var$ of variables, and a set $F$ of variables, a function 
	
	$$
	\theta: (2^{F \cup \forall(\qntPref)})^{\omega} \to (2^{\exists(\qntPref)})^{\omega}
	$$
	
	\noindent is called Skolem function over $(\qntPref, F)$ if, for all $\pi_1, \pi_2 \in (2^{\forall(\qntPref)})^{\omega}$ and $Y \in \exists(\qntPref)$, it holds that
	$${\pi_{1}}_{\rst \Dep^{F}(Y)} = {\pi_{2}}_{\rst \Dep^{F}(Y)} \Rightarrow \theta(\pi_1)_{\rst Y} = \theta(\pi_2)_{\rst Y}\text{.}$$
\end{definition}

Informally, a Skolem function takes interpretations of the variables in $F \cup \forall(\qntPref)$ to return interpretations of the existentially quantified ones in a functional way.
Sometimes, to simplify the notation, we identify $\theta(\pi)$ with $\pi \Cup \theta(\pi)$, that is, $\theta$ \emph{extends} the interpretation $\pi$ to the existentially quantified variables of $\qntPref$.

Skolem functions can be used to define another semantics in \QLTL formulas in prenex normal form.

\begin{definition}[Skolem semantics]
	\label{dfn:skolem-semantics}
	A \QLTL formula in prenex normal form $\varphi = \qntPref \psi$ is \emph{Skolem true} over an interpretation $\pi$ at an instant $i$, written $\pi, i \smodels \varphi$, if there exists a Skolem function $\theta$ over $(\qntPref, \free(\varphi))$ such that $\theta(\pi \Cup \pi_{\forall(\qntPref)}), i \cmodels \psi$ 
\end{definition}

Intuitively, the Skolem semantics characterizes the truth of a \QLTL formula with the existence of a Skolem function that returns the interpretations of the existential quantifications in function of the variables to which they depend.

In principle, there might be formulas $\varphi$ and interpretations $\pi$ such that $\pi \smodels \varphi$ and $\pi \smodels \neg \varphi$, as the Skolem semantics require the existence of two Skolem functions that are defined over different domains, and so not necessarily inconsistent with each other.
However, as the following theorem shows, the Skolem semantics is equivalent to the classic one.
Therefore, for every formula $\varphi$ and an interpretation $\pi$, it holds that $\pi \smodels \varphi$ iff $\pi \not\smodels \neg \varphi$.

\begin{theorem}
	\label{thm:skolemization}
	
	For every \QLTL formula in prenex normal form $\varphi = \qntPref \psi$ and an interpretation $\pi \in (2^{F})^{\omega}$ over the free variables $F = \free(\varphi)$ of $\varphi$ it holds that
%
		$\pi \cmodels \varphi$ if, and only if, $\pi \smodels \varphi$
%
\end{theorem}

\begin{proof}
	Recall taht $\pi \smodels \varphi$ iff there exists a Skolem function $\theta$ over $(\qntPref, F)$ such that, for each interpretation $\pi' \in (2^{\forall(\qntPref)})^{\omega}$, it holds that $\theta(\pi \Cup \pi') \cmodels \psi$.
	
	The proof proceeds by induction on the length of $\qntPref$.
	For the case $\card{\qntPref} = 0$, and so $\qntPref = \epsilon$ , and so that we have that $\varphi = \psi$.
	Moreover, the only Skolem function possible is the identity function over the free variables of $\varphi$, which means that $\pi = \theta(\pi)$ and implies $\pi \cmodels \varphi$ iff $\pi \cmodels \psi$ iff $\theta(\pi) \cmodels \psi$ iff $\pi \smodels \varphi$, an so the statement holds in both directions.
	For the inductive case, we prove the two directions separately.
	
	From the left to right direction, assume that $\pi \cmodels \qntPref \psi$.
	We distinguish two cases.
	
	\begin{itemize}
		\item
		$\qntPref = \exists X \qntPref'$.
		Thus, there exists an interpretation $\pi_{X} \in (2^{X})^{\omega}$ such that $\pi \Cup \pi_{X} \cmodels \qntPref' \psi$.
		By induction hypothesis, it holds that $\pi \Cup \pi_{X} \smodels \qntPref' \psi$ and so there exists a Skolem function $\theta$ over $(\qntPref', F \cup \{X\})$ such that $\theta(\pi \Cup \pi_{X} \Cup \pi') \cmodels \psi$ for each $\pi' \in (2^{\forall(\qntPref')})^{\omega}$.
		Now, observe that $\forall(\qntPref) = \forall(\qntPref')$ and so consider the function $\theta$ is also a Skolem function over $(\qntPref, F)$.
		Hence $\theta(\pi \Cup \pi_{X} \Cup \pi') \cmodels \psi$ for every $\pi'$, which implies that $\pi \smodels \varphi$ and proves the statement.
		
		\item
		$\qntPref = \forall X \qntPref'$.
		Then, for every $\pi_{X}$, it holds that $\pi \Cup \pi_{X} \cmodels \qntPref' \psi$.
		By induction hypothesis, we have that $\pi \Cup \pi_{X} \smodels \qntPref' \psi$ and so there exists a Skolem function $\theta_{\pi_{X}}$ over $(\qntPref', F \cup \{X\})$ such that $\theta_{\pi_{X}}(\pi \Cup \pi_{X} \Cup \pi') \cmodels \psi$ for every $\pi' \in (2^{\forall(\qntPref')})^{\omega}$.
		Now, consider the function $\theta: (2^{F \cup \forall(\qntPref)})^{\omega} \to (2^{\Var(\qntPref)})^{\omega}$ such that $\theta(\pi \Cup \pi') = \theta_{\pi_{X}}(\pi \Cup \pi')$ for each $\pi' \in (2^{\forall(\qntPref)})^{\omega}$.
		Clearly, $\theta$ is a Skolem function over $(\qntPref, F)$.
		Moreover, by its definition, it holds that $\theta(\pi \Cup \pi') \cmodels \psi$ for every $\pi'$, which means that $\pi \smodels \varphi$ and proves the statement.
	\end{itemize}
	
	For the right to left direction, we assume that $\pi \smodels \varphi$ and so that there exists a Skolem function $\theta$ over $(\qntPref, F)$ such that $\theta(\pi \Cup \pi') \cmodels \psi$ for each $\pi' \in (2^{\forall(\qntPref)})^{\omega}$.
	We distinguish two cases.
	
	\begin{itemize}
		\item
		$\qntPref = \exists X \qntPref'$.
		Observe that, since $\Dep^{F}(X) = F$, it holds that $\theta(\pi \Cup \pi')(X) = \theta(\pi \Cup \pi'')(X)$ for every $\pi', \pi''$ and call such interpretation $\pi_{X}$.
		Now, define the Skolem function $\theta'$ over $(\qntPref', F \cup \{X\})$ as $\theta'(\pi \cup \pi') = \theta(\pi \Cup \pi')_{\rst - X}$, that is, the restriction of $\theta$ with the interpretation over $X$ being projected out.
		It holds that $\theta'(\pi \Cup \pi_{X} \Cup \pi') = \theta(\pi \Cup \pi')$ and so that $\theta'(\pi \Cup \pi_{X} \Cup \pi') \cmodels \psi$.
		By induction hypothesis, we have that $\pi \Cup \pi_{X} \smodels \qntPref' \psi$, which in turns implies that $\pi \cmodels \exists X \qntPref' \psi$ and so that $\pi \cmodels \qntPref \psi$, which proves the statement.
		
		\item
		$\qntPref = \forall X \qntPref'$.
		Note that $\forall(\qntPref) = \forall(\qntPref') \cup \{ X \}$, and so that $\theta$ is also a Skolem function over $(\qntPref', F \cup \{X\})$.
		By induction hypothesis, we obtain that, for every $\pi_{X}$, it holds that $\theta(\pi \Cup \pi_{X} \cup \pi') \cmodels \psi$ implies that $\pi \Cup \pi_{X} \cmodels \qntPref' \psi$ for every $\pi_{X}$, which means that $\pi \cmodels \forall X \qntPref' \psi$, and so that $\pi \cmodels \qntPref \psi$, and then the statement is proved.
	\end{itemize}
\end{proof}

\section{Behavioral QLTL}
\label{sec:behavioral-QLTL}

The classic semantics of \QLTL requires to consider at once the evaluation of the variables on the whole trace.
This gives rise to counter-intuitive phenomena.
Consider the formula $\forall x \exists y (\G x \leftrightarrow y)$.
Such a formula is satisfiable.
Indeed, on the one hand, for the interpretation assigning always true to $x$, the interpretation that makes $y$ true at the beginning satisfies the temporal part.
On the other hand, for every other interpretation making $x$ false sometimes, the interpretation that makes $y$ false at the beginning satisfies the temporal part.
However, in order to correctly interpret $y$ on the first instant, one needs to know in advance the entire interpretation of $x$.
Such requirement is practically impossible to fulfill and does not reflect the notion of \emph{reactive systems}, where the output of system variables at the $k$-th instant of the computation depends only on the past assignments of the environment variables.
Such  principle is often referred as \emph{behavioral} principle in the context of strategic reasoning, see e.g., \cite{MMPV14,GBM20}.

Here, we propose two alternative semantics for \QLTL, which are of interest when \QLTL is used in the context of strategic reasoning and planning.
Indeed there we require strategies to be processes in the sense of \cite{ALW89}, i.e., the next move depends only on the past history and the current situation.
The two semantics are inspired by two different contexts of planning and distributed synthesis.
The first regards \emph{partial controllability} with \emph{partial observability}, in which a process in a distributed architecture controls part of the system variables and assigns their value according to the past and present values of the environment variables that are made visible to it.
The second regards \emph{partial controllability} with \emph{full observability}, in which the process can base its choices according to the past evaluation of all variables and the present evaluation of the depending ones.

To formally define the two semantics we exploit two different forms of Skolem functions, each of them producing different effects on the notion of formula satisfaction.
These definitions take into account the reactive feature of dependency discussed above.
In addition, we prove their connection with the classic notion of strategy as intended in synthesis and distributed synthesis~\cite{PR89,KV01,FS05}.
In the next subsections, we introduce these two semantics and discuss their relationship with the classic semantics of \QLTL.
Subsequently, we show their connection with the synthesis problem of the corresponding contexts.

\subsection{Behavioral semantics}
\label{subsec:behavioral-semantics}

We now introduce \emph{behavioral \QLTL}, denoted \BQLTL, a logic with the same syntax as of prenex normal form \QLTL but where the semantics is defined in terms of behavioral Skolem functions: a modified version of the Skolem functions introduced in the previous section.

\begin{definition}[Behavioral Skolem function]
	\label{def:behavioral-skolem}
	For a given quantification prefix $\qntPref$ defined over a set $\Var(\qntPref) \subseteq \Var$ of propositional variables and a set $F$ of variables not occurring in $\qntPref$, a Skolem function $\theta$ over $(\qntPref, F)$ is \emph{behavioral} if , for all $\pi_1, \pi_2 \in (2^{F \cup \forall(\qntPref)})^{\omega}$, $k \in \SetN$, and $X \in \exists(\qntPref)$, it holds that
	\begin{center}
		$\pi_{1}(0,k)_{\rst \Dep^{F}(X)} = \pi_{2}(0,k)_{\rst \Dep^{F}(X)}$ implies $\theta(\pi_{1})_{\rst X} = \theta(\pi_{2})_{\rst X}$.
	\end{center}
\end{definition}

The behavioral Skolem functions capture the fact that the interpretation of existentially quantified variables depend only on the past and present values of free and universally quantified variables.
This offers a way to formalize the semantics of \BQLTL as follows.

\begin{definition}
	\label{def:BQLTL-semantics}
	
	A \BQLTL formula $\varphi = \qntPref \psi$ is true over an interpretation $\pi$ in an instant $i$, written $\pi, i \bmodels \qntPref \psi$, if there exists a behavioral Skolem function $\theta$ over $(\qntPref, \free(\varphi))$ such that $\theta(\pi \Cup \pi'), i \cmodels \psi$ for every $\pi' \in (2^{F \cup \forall(\qntPref)})^{\omega}$.
	
%
\end{definition}

	A \BQLTL formula $\varphi$ is true on an interpretation $\pi$, written $\pi \bmodels \varphi$, if $\pi, 0 \bmodels \varphi$.
	A formula $\varphi$ is \emph{satisfiable} if it is true on some interpretation and \emph{valid} if it is true in every interpretation.

	Clearly, since \BQLTL shares the syntax with \QLTL, all the definitions that involve syntactic elements, such as free variables and alternation, apply to this variant the same way.

	As for \QLTL, the satisfiability of a \BQLTL formula $\varphi$ is equivalent to the one of $\exists \free(\varphi) \varphi$, as well as the validity is equivalent to the one of $\forall \free(\varphi) \varphi$.
	However, the proof of this is not as straightforward as for the classic semantics case.
\begin{theorem}
	\label{thm:behavioral-satisfiability}
	For every \BQLTL formula $\varphi = \qntPref \psi$, it holds that $\varphi$ is satisfiable if, and only if, $\exists \free(\varphi) \varphi$ is satisfiable.
	Moreover, $\varphi$ is valid if, and only if, $\forall \free(\varphi) \varphi$ is valid.
\end{theorem}

\begin{proof}
	We show the proof only for satisfiability, as the one for validity is similar.
	The proof proceeds by double implication.
	From left to right, assume that $\varphi$ is satisfiable, therefore there exists an interpretation $\pi$ over $F = \free(\varphi)$ such that $\pi \bmodels \varphi$, which in turns implies that there exists a behavioral Skolem function $\theta$ over $(\qntPref, F)$ such that $\theta(\pi \Cup \pi') \cmodels \psi$ for every interpretation $\pi' \in (2^{\forall(\qntPref)})^{\omega}$.
	Consider the function $\theta': (2^{\forall(\qntPref)})^{\omega} \to (2^{\exists(\qntPref) \cup F})^{\omega}$ defined as $\theta'(\pi') = \theta(\pi \Cup \pi') \Cup \pi$, for every $\pi' \in (2^{\forall(\qntPref)})^{\omega}$.
	Clearly, it is a behavioral Skolem function over $(\exists F \qntPref, \emptyset)$ such that $\theta'(\pi') \models \psi$ for every $\pi' \in (2^{\forall(\qntPref)})^{\omega}$, which implies that $\exists F \varphi$ is satisfiable.
	From right to left, the reasoning is similar and left to the reader.
\end{proof}

Note that every behavioral Skolem function is also a Skolem function.

This means that a formula $\varphi$ interpreted as \BQLTL is true on $\pi$ implies that the same formula is true on $\pi$ also when it is interpreted as \QLTL.
The reverse, however, is not true.
Consider again the formula $\varphi = \forall x \exists y (\G x \leftrightarrow y)$.
We have already shown that this is satisfiable when interpreted as \QLTL. However, it is not satisfiable as a \BQLTL formula.

\begin{lemma}
	\label{lmm:behavioral-implies-classic}
	
	For every \BQLTL formula $\varphi$ and an interpretation $\pi$ over the set $\free(\varphi)$ of free variables, if $\pi \bmodels\varphi$ then $\pi \cmodels \varphi$.
	On the other hand, there exists a formula $\varphi$ and an interpretation $\pi$ such that $\pi \cmodels \varphi$ but not $\pi \bmodels \varphi$.
\end{lemma}

\begin{proof}
	The first part of the theorem follows from the fact that every behavioral Skolem function is also a Skolem function and so, if $\pi \bmodels \varphi$, clearly also $\pi \smodels \varphi$ and so, from Theorem~\ref{thm:skolemization}, that $\pi \cmodels \varphi$.
	
	For the second part, consider the formula $\varphi = \forall x \exists y (\G x \leftrightarrow y)$.
	We have already shown that such formula is satisfiable.
	However, it is not behavioral satisfiable.
	Indeed, assume by contradiction that it is behavioral satisfiable and let $\theta$ the behavioral Skolem function such that $\theta \cmodels (\G x \leftrightarrow y)$.
	Now consider two interpretations $\pi_{1}$ over $x$ that always assigns true, and $\pi_{2}$ that assigns true on $x$ at the first iteration and then always false.
	It holds that $\pi_{1}(0) = \pi_{2}(0)$ and therefore, since $x \in \Dep(y)$ and $\theta$ is behavioral, it must be the case that $\theta(\pi_1)(0)^{\rst y} = \theta(\pi_2)(0)^{\rst y}$.
	Now, if such value is $\theta(\pi_1)(0)^{\rst y} = \false$, then it holds that $\theta(\pi_{1}) \not\cmodels \G x \leftrightarrow y)$.
	On the other hand, if $\theta(\pi_2)(0)^{\rst y} = \true$, then it holds that $\theta(\pi_{2}) \not\cmodels \G x \leftrightarrow y)$, which means that $\theta \not\cmodels (\G x \leftrightarrow y)$, a contradiction.
\end{proof}

Lemma~\ref{lmm:behavioral-implies-classic} has implications also on the meaning of negation in \BQLTL.
Indeed, both the formula $\varphi = \forall x \exists y (\G x \leftrightarrow y)$ and its negation are not satisfiable, that is $\not\bmodels \varphi$ and $\not\bmodels \neg \varphi$~\footnote{Note that, being $\varphi$ with no free variables, we can omit the interpretation $\pi$ as the only possible is the empty one}.
This is a common phenomenon, as it also happens when considering the behavioral semantics of logic for the strategic reasoning~\cite{MMPV14,GBM20}.
It is important, however, to notice that there are three syntactic fragments for which \QLTL and \BQLTL are equivalent.
Precisely, the fragments $\Pi_{0}^{\BQLTL}$, $\Sigma_{0}^{\BQLTL}$, and $\Sigma_{1}^{\BQLTL}$.
The reason is that the sets of Skolem and behavioral Skolem functions for these formulas coincide, and so the existence of one implies the existence of the other.

\begin{theorem}
	\label{thm:no-behavioral-dep-equivalence}
	For every \BQLTL formula $\varphi = \qntPref \psi$ in the fragments $\Pi_{0}^{\BQLTL}$, $\Sigma_{0}^{\BQLTL}$, and $\Sigma_{1}^{\BQLTL}$ and an interpretation $\pi$, it holds that  $\pi \bmodels \varphi$ if, and only if, $\pi \smodels \varphi$.
\end{theorem}

\begin{proof}
	The proof proceeds by double implication.
	From left to right, it follows from Lemma~\ref{lmm:behavioral-implies-classic}.
	From right to left, consider first the case that $\varphi \in \Pi_{0}^{\QLTL}$.
	Observe that $\exists(\qntPref) = \emptyset$ and so the only possible Skolem function $\theta$ returns the empty interpretation on every possible interpretation $\pi \Cup \pi' \in (2^{\free(\varphi) \cup \forall(\qntPref)})^{\omega}$.
	Such Skolem function is trivially behavioral and so we have that $\pi \smodels \varphi$ implies $\pi \bmodels \varphi$. \\
	
	For the case of $\varphi \in \Sigma_{0}^{\QLTL} \cup \Sigma_{1}^{\QLTL}$, assume that $\pi, \smodels \varphi$ and let $\theta$ be a Skolem function such that $\theta(\pi \cup \pi') \cmodels \varphi$ for every $\pi' \in (2^{\forall(\qntPref)})^{\omega}$.
	Observe that, for every $Y \in \exists(\qntPref)$, it holds that $\Dep_{\qntPref} = \emptyset$ and so the values of $Y$ depend only on the free variables in $\varphi$.	
	Now, consider the Skolem function $\theta'$ over $(\qntPref, \free(\varphi))$ defined such that as $\theta'(\pi') \defeq \theta(\pi'_{\rst \forall(\qntPref) \Cup \pi})$.
	As $\theta$ is a Skolem function and $\Dep_{\qntPref} = \emptyset$, it holds that $\theta'(\pi')(Y) = \theta'(\pi'')(Y)$ for every $\pi', \pi'' \in (2^{\forall(\qntPref)})^{\omega}$ and so $\theta'$ is trivially behavioral.
	Moreover, from its definition, it holds that $\theta'(\pi \Cup \pi') \cmodels \psi$ for every $\pi' \in (2^{\forall(\qntPref)})^{\omega}$, which implies $\pi \bmodels \varphi$.
\end{proof}

Theorem~\ref{thm:no-behavioral-dep-equivalence} shows that for these three fragments of \BQLTL, the satisfiability problem can be solved by employing \QLTL satisfiability.
This also comes with the same complexity, as we just interpret the \BQLTL formula directly as \QLTL one.

\begin{corollary}
	\label{cor:no-behavioral-fragments-complexity}
	The satisfiability problem for the fragments $\Pi_{0}^{\BQLTL}$ and $\Sigma_{0}^{\BQLTL}$ is PSPACE-complete.
	Moreover, the satisfiability problem for the fragment $\Sigma_{1}^{\BQLTL}$ is EXPSPACE-complete.
\end{corollary}

	\subsection{Behavioral QLTL Satisfiability}

	We now turn into solving the satisfiability problem for \BQLTL formulas that are not in fragments $\Pi_{0}^{\BQLTL}$, $\Sigma_{0}^{\BQLTL}$, and $\Sigma_{1}^{\BQLTL}$.
	Analogously to the case of \QLTL, note that Theorem~\ref{thm:behavioral-satisfiability} allows to restrict our attention to closed formulas.
	We use an automata-theoretic approach inspired by the one employed in the synthesis of distributed systems~\cite{KV01,FS05,Sch08a}.
	This requires some definitions and results, presented below.
	
	For a given set $\Upsilon$ of directions the \emph{$\Upsilon$-tree} is the set $\Upsilon^{*}$ of finite words.
	The elements of $\Upsilon^{*}$ are called nodes, and the empty word $\varepsilon$ is called \emph{root}.
	For every $x \in \Upsilon^{*}$, the nodes $x \cdot c \in \Upsilon^{*}$ are called children.
	We say that $c = \dir(x \cdot c)$ is the direction of the node $x \cdot c$, and we fix some $\dir(\varepsilon) = c_0 \in \Upsilon$ to be the direction of the root.
	Given two finite sets $\Upsilon$ and $\Sigma$, a $\Sigma$-labeled $\Upsilon$-tree is a pair $\tuple{\Upsilon^{*}, l}$ where $l: \Upsilon^{*} \to \Sigma$ maps/labels every node of $\Upsilon^{*}$ into a letter in $\Sigma$.
	
	For a set $\Theta \times \Upsilon$ of directions and a node $x \in (\Theta \times \Upsilon)^{*}$, $\hide_{\Upsilon}(x)$ denotes the node in $\Theta^{*}$ obtained from $x$ by replacing $(\vartheta, \upsilon)$ with $\vartheta$ in each letter of $x$.
	The function $\xray_{\Xi}$ maps a $\Sigma$-labeled $(\Xi \times \Upsilon)$-tree $\tuple{(\Xi \times \Upsilon)^{*}, l}$ into a $\Xi \times \Sigma$-labeled $(\Xi \times \Upsilon)$-tree $\tuple{(\Xi \times \Upsilon)^{*}, l'}$ where $l'(x) = (pr_1(\dir(x)), l(x))$ adds the $\Xi$-direction of $x$ to its labeling.

	An \emph{alternating automaton} $\Automaton = (\Sigma, Q, q_0, \delta, \alpha)$ runs over $\Sigma$-labeled $\Upsilon$-trees (for a predefined set of directions $\Upsilon$).
	The set of states $Q$ is finite with $q_0$ being a designated initial state, while $\delta: Q \times \Sigma \to \Bool^{+}(Q \times \Upsilon)$ denotes a transition function, returning a positive Boolean formula over pairs of states and directions, and $\alpha$ is an acceptance condition.
	
	We say that $\Automaton$ is \emph{nondeterministic}, and denote it with the symbol $\NAutomaton$, if every transition returns a positive Boolean formula with only disjunctions.
	Moreover, was that it is deterministic \emph{deterministic}, and denote it with the symbol $\DAutomaton$, if every transition returns a single state.

	A run tree of $\Automaton$ on a $\Sigma$-labeled $\Upsilon$ tree $\tuple{\Upsilon^{*}, l}$ is a $Q \times \Upsilon$-labeled tree where the root is labeled with $(q_0, l(\varepsilon))$ and where, for a node $x$ with a label $(q,x)$, and a set of children $\child(x)$, the labels of these children have the following properties:
	
	\begin{itemize}
		\item
		for all $y \in \child(x)$, the label of $y$ is of the form $(q_y, x \cdot c_y)$ such that $(q_y, c_y)$ is an atom of the formula $\delta(q, l(x))$ and
		\item
		the set of atoms defined by the children of $x$ satisfies $\delta(q, l(x))$.
	\end{itemize}
	
	We say that $\alpha$ is a \emph{parity} condition if it is a function $\alpha: Q \to C (\subset \SetN)$ mapping every state to a natural number, sometimes referred as color.
	Alternatively, it is a \emph{Streett} condition if it is a set of pairs $\{(G_i, R_i)\}_{i \in I}$, where each $G_i, R_i$ is a subset of $Q$.
	An infinite path $\rho$ over $Q$ fulfills a parity condition $\alpha$ if the highest color of mapped by $\alpha$ over $\rho$ that appears infinitely often is even.
	The path $\rho$ fulfills a Streett condition if for every $i \in I$, either an element of $G_i$ or no element of $R_i$ occurs infinitely often on $\rho$.
	A run tree is \emph{accepting} if all its path fulfill the acceptance condition $\alpha$.
	A tree is accepted by $\Automaton$ if there is an accepting tree run over it.
	By $\Lang(\Automaton)$ we denote the set of trees accepted by $\Automaton$.
	An automaton $\Automaton$ is \emph{empty} if $\Lang(\Automaton) = \emptyset$.
		
	For a $\Sigma$-labeled $\Upsilon$-tree $\tuple{\Upsilon^{*}, l_{\Sigma}}$ and a $\Xi$-labeled $\Upsilon \times \Theta$-tree $\tuple{(\Upsilon \times \Theta)^{*}, l_{\Xi}}$, their \emph{composition}, denoted $\tuple{\Upsilon^{*}, l_{\Sigma}} \oplus \tuple{(\Upsilon \times \Theta)^{*}, l_{\Xi}}$ is the $\Xi \times \Sigma$-labeled $\Upsilon \times \Theta$-tree $\tuple{(\Upsilon \times \Theta)^{*}, l}$ such that, for every $x \in (\Upsilon \times \Theta)^{*}$, it holds that $l(x) = l_{\Xi}(x) \cup l_{\Sigma}(\hide_{\Theta}(x))$.
	Observe that the $\Upsilon$-component appears in both the trees.
	Their composition, indeed, can be seen as an extension of the labeling $l_{\Xi}$ with the labeling $l_{\Sigma}$ in a way that the choices for it are oblivious to the $\Theta$-component of the direction.
	A more general definition of tree composition is given in~\cite{FS05} where the $\Sigma$-labeling is included as a direction and made consistent with it by means of an $\xray$ operation.
	
	For a set $\Tree$ of $\Xi \times \Sigma$-labeled $\Upsilon \times \Theta$-trees, $\shape_{\Xi,\Upsilon}(\Tree)$ is the set of $\Sigma$-labeled $\Upsilon$-trees $\tuple{\Upsilon^{*}, l_{\Sigma}}$ for which there exists a $\Xi$-labeled $\Upsilon \times \Theta$-tree $\tuple{(\Upsilon \times \Theta)^{*}, l_{\Xi}}$ such that $\tuple{\Upsilon^{*}, l_{\Sigma}} \oplus \tuple{(\Upsilon \times \Theta)^{*}, l_{\Xi}} \in \Tree$.
	Intuitively, the shape operation performs a nondeterministic guess on the $\Sigma$-component of the trees by taking into account only the $\Upsilon$-component of the directions.
	This allows to refine the set of trees into those ones for which a decomposition consistent with this limited dependence is possible.
	Interestingly, being this nondeterministic guess similar to an existential projection, we can also refine a (nondeterministic) parity tree automaton $\NAutomaton$ in order to recognize the shape operation of its language.
	Indeed, consider a nondeterministic parity tree automaton $\NAutomaton = (\Xi \times \Sigma, Q, q_0, \delta, \alpha)$ recognizing $\Xi \times \Sigma$-labeled $\Upsilon$-trees, the automaton $\change_{\Xi,\Upsilon}(\NAutomaton) = (\Sigma, Q, q_0, \delta', \alpha)$ recognizes $\Sigma$-labeled $\Upsilon$-trees where
	
	\begin{center}
		$\delta'(q, \sigma) = \bigvee_{\xi \in \Xi, f \in \delta(q, (\xi, \sigma))} \bigwedge_{\upsilon \in \Upsilon, \vartheta \in \Theta}(f(\upsilon, \vartheta), (\xi, \sigma))$.
	\end{center}
	
	Intuitively, the automaton $\change_{\Xi,\Upsilon}(\NAutomaton)$ encapsulates and then nondeterministically guesses $\Xi$-labeled $\Upsilon \times \Theta$-trees in a way that their composition with the read $\Sigma$-labeled $\Upsilon$-tree is accepted by $\NAutomaton$.
	The following holds.
	
	\begin{theorem}{\cite[Theorem 4.11]{FS05}}
		\label{thm:shape-automaton}
		For every nondeterministic parity tree automaton $\NAutomaton$ over $\Xi \times \Sigma$-labeled $\Upsilon \times \Theta$-trees, it hols that $\Lang(\change_{\Xi,\Upsilon}(\NAutomaton)) = \shape_{\Xi,\Upsilon}(\Lang(\NAutomaton))$.
	\end{theorem}
	
	We can apply the change operation only on nondeterministic automata.
	This means that, in order to recognize the shape language of a parity alternating automaton $\Automaton$, we first need to turn it into a nondeterministic one.
	This can be done by means of two steps: we first turn $\Automaton$ into a nondeterministic Street automaton $\NAutomaton_{S}$ that recognize the same language $\Lang(\NAutomaton_{S}) = \Lang(\Automaton)$, and then turn it into a nondeterministic parity $\NAutomaton$ such that $\Lang(\NAutomaton) = \Lang(\NAutomaton_{S}) = \Lang(\Automaton)$.
	If $\Automaton$ has $n = \card{Q}$ states and $c = \card{C}$ colors, then the automaton $\NAutomaton_{S}$ has $n^{O(c \cdot n)}$ states and $O(c \cdot n)$ pairs such that $\Lang(\Automaton) = \Lang(\NAutomaton_{S})$~\cite{MS84}.
	In addition, it the nondeterministic Street automaton $\NAutomaton_{S}$ has $m$ states and $p$ pairs, we can build a nondeterministic parity automaton $\NAutomaton$ with $p^{O(p)} \cdot m$ states and $O(p)$ colors~\cite{FS05}.
	By applying these two constructions, we then transform an alternating parity automaton $\Automaton$ into a nondeterministic one $\NAutomaton$ accepting the same tree-language.
	Note that $\NAutomaton$ is of size single exponential with respect to $\Automaton$.
	Indeed, we obtain it with $n' = O(c \cdot n)^{O(c \cdot n)} = n^{O(c \cdot n)}$ states~\footnote{The last equivalence because the number $c$ of colors is bounded by the number $n$ of states.} and $c' = O(c \cdot n)$ colors.
	By $\ndet(\Automaton) = \NAutomaton$ we denote the transformation of an alternating parity automaton into a nondeterministic parity one.

	From now on, we consider closed \BQLTL formulas being of the form $\qntPref \psi = \exists Y_1 \forall X_1 \allowbreak \ldots \exists Y_n \forall X_n \psi$ with $Y_1$ and $X_n$ being possibly empty.
	Therefore, we refer to $\theta$ as a behavioral Skolem function over $\qntPref$, as the set $F = \emptyset$ is always empty.
	Moreover, we define $\hat{X_{i}} = \bigcup_{j \leq i} X_i$ and $\hat{Y_{i}} = \bigcup_{j \leq i} Y_i$, with $X = \hat{X_{0}}$ and $Y = \hat{Y_{0}}$, respectively.
	Finally, we define $\check{X_{i}} = \bigcup_{j > i} X_i$ and $\check{Y_{i}} = \bigcup_{j > i} Y_i$, respectively.

	A behavioral Skolem function $\theta$ over $\qntPref$ can be regarded as the labeling function of a $2^{Y}$-labeled $2^{X}$-tree.
	In addition, such labeling fulfills a \emph{compositional} property, as it is expressed in the following lemma.
	
	\begin{lemma}
		\label{lmm:bijection-behavioral-trees}
		Let $\qntPref = \exists Y_1 \forall X_1 \ldots \exists Y_n \forall X_n$ be a prefix quantifier.
		A $2^{Y}$-labeled $2^{X}$-tree $\theta$ is a behavioral Skolem function over $\qntPref$ iff there exist a tuple $\theta_1, \ldots, \theta_n$,  where $\theta_i$ is a $2^{Y_i}$-labeled $2^{\hat{X_{i}}}$-tree, such that 
		$\theta = \theta_1 \oplus \ldots \oplus \theta_n$.		
	\end{lemma}
	
	\begin{proof}
		The proof proceeds by double implication.
		From left to right, consider a behavioral Skolem function $\theta$ and, for every $1 \leq i \leq n$, consider the $2^{Y_i}$-labeled $2^{\hat{X_{i}}}$-tree $\theta_i$, defined as $\theta_{i}(x) = \theta(x \times x')_{\rst Y_{i}}$ where $x \in (2^{\hat{X_{i}}})^{*}$ and $x' \in (2^{X \setminus \hat{X_{i}}})^{*}$.
		Note that $\Dep_{\qntPref}(Y_{i}) = \hat{X_{i}}$ and so the definition of $\theta_{i}$ over $x$ does not really depend on the values in $x'$, therefore it is well-defined.
		By applying the definition of tree composition, it easily follows that $\theta = \theta_1 \oplus \ldots \oplus \theta_n$.
		
		For the right to left direction, let $\theta_1, \ldots, \theta_{n}$ be labeled trees and consider the composition $\theta = \theta_1 \oplus \ldots \oplus \theta_n$.
		From the definition of tree composition, it follows that for every $i$, $\theta(x)_{\rst Y_{i}} = \theta_{i}(x_{\hat{X_{i}}})$, which fulfills the requirement for $\theta$ of being a behavioral Skolem function over $\qntPref$.
	\end{proof}

	We now show how to solve the satisfiability problem for \BQLTL with an automata theoretic approach.
	To do this, we first introduce some notation.
	For a list of variables $(Y_i, X_i)$, consider the quantification prefix $\check{\qntPref_{i}} \defeq \forall \check{X_{i}} \exists \check{Y_{i}}$ and then the quantification prefix $\qntPref_{i} \defeq \exists Y_1 \forall X_1 \ldots \exists Y_i \forall X_i \check{\qntPref_{i}}$.
	Intuitively, every quantification prefix $\qntPref_{i + 1}$ is obtained from $\qntPref_{i}$ by pulling the existential quantification of $Y_{i + 1}$ up before the universal quantification of $X_{i + 1}$.
	Clearly, we obtain that $\qntPref_{0} = \forall X \exists Y$ and $\qntPref_{n} = \qntPref$.
	The automata construction builds on top of this quantifier transformation.
	First, recall that the satisfiability of $\qntPref_{0} \psi$ amounts to solving the synthesis problem for $\psi$ with $X$ and $Y$ being the set of variables controlled by the environment and the system, respectively.
	Let $\Automaton_{0}$ be an alternating parity automaton that solves the synthesis problem, thus $2^{Y}$-labeled $2^{X}$-trees representing the models of $\psi$.
	Now, for every $i < n$, define $\Automaton_{i + 1} \defeq \change_{2^{Y_i}, 2^{\hat{X_{i}}}}(\ndet(\Automaton_{i}))$.
	We have the following.
	
	\begin{theorem}
		\label{thm:automata-sequence}
		For every $i \leq n$, the formula $\qntPref_{i} \psi$ is satisfiable iff $\Lang(\Automaton_{i}) \neq \emptyset$, where 
		
		\begin{itemize}
			\item
				$\Automaton_{0}$ is the alternating parity automaton that solves the synthesis problem for $\psi$ with system variables $Y$ and environment variables $X$, and
				
			\item
				$\Automaton_{i + 1} \defeq \change_{2^{Y_i}, 2^{\hat{X_{i}}}}(\ndet(\Automaton_{i}))$, for every $i < n$.
		\end{itemize}
		
	\end{theorem}

	\begin{proof}
	We prove the theorem by induction through a stronger statement.
	We show that the automaton $\Automaton_{i}$ accepts $2^{\check{Y_{i}}}$-labeled $2^{X}$-trees $\theta_{i}$ for which there exists a sequence $\theta_{1}, \ldots, \theta_{i-1}$ such that $\theta_1 \oplus \ldots \oplus \theta_{i}$ is a behavioral Skolem function over $\qntPref_{i}$ that satisfies $\qntPref_{i} \psi$.
	
	For the base case, the statement boils down to the fact that the automaton $\Automaton_{0}$ accepts the $2^{Y}$-labeled $2^{X}$-trees that solve the synthesis problem for $\psi$.
	
	For the induction case, assume that the statement is true for some $i$.
	Thus, the automaton $\Automaton_{i}$, and then its nondeterministic version $\ndet(\Automaton_{i})$ accept $2^{\check{Y_{i}}}$-labeled $2^{X}$-trees $\theta_{i}$ for which there exists a sequence $\theta_{1}, \ldots, \theta_{i-1}$ such that $\theta_1 \oplus \ldots \oplus \theta_{i}$ is a behavioral Skolem function that satisfies $\qntPref_{i} \psi$.
	Now, consider the automaton $\Automaton_{i + 1} = \change_{2^{Y_i}, 2^{\hat{X_{i}}}}(\ndet(\Automaton_{i}))$.
	From Theorem~\ref{thm:shape-automaton}, it holds that it accepts $2^{\check{Y_{i + 1}}}$-labeled $2^{X}$-trees $\theta_{i + 1}$ that are in $\shape_{2^{Y_i + 1}, 2^{\hat{X_{i+1}}}}(\Lang(\Automaton_{i}))$ and so for which there exists a $2^{Y_i}$-labeled $2^{\hat{X_{i+1}}}$-tree $\theta_{i}'$ such that $\theta_{i}' \oplus \theta_{i + 1} \in \Lang(\Automaton_{i})$.
	Observe that now the variables $Y_{i}$ are handled over a $2^{\hat{X_{i}}}$-tree and so they do not depend on variables in $\check{X_{i}}$ anymore.
	This implies that the composition $\theta_1 \oplus \theta_{i - 1} \oplus \theta'_{i} \oplus \theta_{i + 1}$ is a behavioral Skolem over $\qntPref_{i + 1}$ that satisfies $\qntPref_{i + 1} \psi$, and the statement is proved.
\end{proof}

	Theorem~\ref{thm:automata-sequence} shows that the automata construction is correct.
	The complexity of solving the satisfiability of \BQLTL is stated below.
		
	\begin{theorem}
		\label{thm:behavioral-satisfiability-complexity}
		The satisfiability problem of a \BQLTL formula of the form $\varphi = \exists Y_1 \forall X_1 \ldots \allowbreak \exists Y_n \forall X_n \psi$ can be solved in $(n+1)$-EXPTIME-complete.
	\end{theorem}	
	
	\begin{proof}
		From Theorem~\ref{thm:automata-sequence}, we reduce the problem to the emptiness of the automaton $\Automaton_{n}$, whose size is $n$-times exponential in the size of $\psi$, as we apply $n$ times the nondeterminisation, starting from the automaton $\Automaton_{\psi}$ that solves the synthesis problem for $\psi$.
		As the emptiness of the alternating parity automaton $\Automaton_{n}$ involves another exponential blow-up, we obtain that the overall procedure is $(n+1)$-EXPTIME.
		
		A matching lower-bound is obtained from the synthesis of distributed synthesis for hierarchically ordered architecture processes with \LTL objectives, presented in~\cite{PR90}, that is $(n+1)$-EXPTIME-complete with $n$ being the number of processes.
		Indeed, every process $p_i$ in such architecture synthesizes a strategy represented by a $2^{O_i}$-labeled $2^{I_i}$-tree, with $O_i$ being the output variables and $I_i$ the input variables.
		An architecture $A$ is hierarchically ordered if $I_i \subseteq I_{i + 1}$, for every process $p_i$.
		Thus, for an ordered architecture $A$ and an \LTL formula $\psi$, consider the variables $Y_{i} = O_{i}$ and $X_{i} = I_i \setminus I_{i -1}$ and the \BQLTL formula $\varphi = \exists Y_1 \forall X_1 \ldots \exists Y_n \forall X_n \psi$.
		A behavioral Skolem function $\theta$ that makes $\varphi$ true corresponds to an implementation for the architecture that realizes $(A, \psi)$.
		Moreover, the satisfiability of $\varphi$ is $(n + 1)$-EXPTIME, matching the lower-bound complexity of the realizability instance.
	\end{proof}

	\section{Weak-Behavioral QLTL}
	\label{sec:weak-behavioral-semantics}
	
	We now introduce \emph{weak-behavioral} \QLTL, denoted \WBQLTL, that can be  used to model systems with \emph{full observability} over the executions history.
	In such system every action is \emph{public}, meaning that it is visible to the entire system once it is occurred.
	In order to model this, we introduce an alternative definition of Skolem function, which we call here \emph{weak-behavioral}.
	We study the satisfiability problem of \WBQLTL and show that its complexity is 2-EXPTIME-complete via a reduction to a Multi-Player Parity Game~\cite{MMS16} with a double exponential number of states and a (single) exponential number of color.
		
	Analogously to the case of \BQLTL, the logic \WBQLTL is defined in a Skolem-based approach.
	
	\begin{definition}
		\label{def:weak-behavioral-skolem}
		
		For a given quantification prefix $\qntPref$ defined over a set $\Var(\qntPref) \subseteq \Var$ of propositional variables and a set $F$ of variables not occurring in $\qntPref$, a function $\theta: (2^{F \cup \forall(\qntPref)})^{\omega} \to (2^{\exists(\qntPref)})^{\omega}$ is a \emph{weak-behavioral Skolem function} over $(\qntPref, \free(\varphi))$ if, for all $\pi_1, \pi_2 \in (2^{F \cup \forall(\qntPref)})^{\omega}$, $k \in \SetN$, and $Y \in \exists(\qntPref)$, it holds that
		\begin{center}
			$\theta(\pi_{1})(0,k) = \theta(\pi_{2})(0,k)$ and $\pi_{1}(k + 1)_{\rst \Dep^{F}(Y)} =  \pi_{2}(k + 1)_{\rst \Dep^{F}(Y)}$ implies $\theta(\pi_{1})_{\rst Y} = \theta(\pi_{2})_{\rst Y}$.
		\end{center}
		
	\end{definition}

	In weak-behavioral Skolem functions, the evaluation of existential variables $Y$ at every instant depends not only on the current evaluation of $\Dep^{F}(Y)$ but also the evaluation history of each variable.
	The semantics of \WBQLTL is given below.

	\begin{definition}
		\label{def:weak-behavioral-semantics}
		
		
		A \WBQLTL formula $\varphi = \qntPref \psi$ is true over an interpretation $\pi$ at an instant $i$, written $\pi, i \wbmodels \varphi$, if there exists a weak-behavioral Skolem function $\theta$ over $(\qntPref, \free(\varphi))$ such that $\theta(\pi \Cup \pi'), i \cmodels \psi$, for every $\pi' \in (2^{F \cup \forall(\qntPref)})^{\omega}$.
		
	\end{definition}
	
	Differently from behavioral, \WBQLTL is not a special case of \QLTL.
	As a matter of fact, they are incomparable.
	Consider again the formula 
	This is due to the fact that the existentially quantified variables depend, for standard Skolem functions, on the future of their dependencies, whereas, weak-behavioral functions, on the whole past of the computation, including the non-dependencies.
	
	Consider again the formula $\varphi = \forall x \exists y (\G x \leftrightarrow y)$.
	This is not satisfiable as a \WBQLTL formula, as this semantics still does not allow existential variables to depend on the future interpretation of the universally quantified ones.
	On the other hand, the formula $\varphi = \exists y \forall x (\F x \leftrightarrow \F y)$ is satisfiable as a \WBQLTL.
	Indeed, the existentially quantified variable $y$ can determine its value on an instant $i$ by looking at the entire history of assignments, including those for $x$, although only on the past but not the present instant $i$ itself.
	However, the semantics of both \QLTL and \BQLTL does not allow such dependence, which makes $\varphi$ non satisfiable as both \QLTL and \BQLTL.
	
	
	\begin{lemma}
		\label{lmm:semantics-comparison}
		
		There exists a satisfiable \BQLTL formula that is not satisfiable as \WBQLTL.
		Moreover, there exists a satisfiable \WBQLTL formula that is not satisfiable as \BQLTL.		
	\end{lemma}
	
	\section{Weak-Behavioral QLTL Satisfiability}
		
	We now address the satisfiability problem for \WBQLTL by showing a reduction to multi-agent parity games~\cite{MMS16}.
	Intuitively, a \WBQLTL formula of the form $\varphi = \qntPref \psi$, with $\psi$ being an \LTL formula, establishes a multi-player parity game with $\psi$ determining the parity acceptance condition and $\qntPref$ setting up the Player's controllability and team side.
	In order to present this result, we need some additional definition.
	
	An $\omega$-word over an alphabet $\Sigma$ is a special case of a $\Sigma$-labeled $\Upsilon$-tree where the set of directions is a singleton.
	Being the set $\Upsilon$ irrelevant, an $\omega$-word is also represented as an infinite sequence over $\Sigma$.
	The tree automata accepting $\omega$-words are also called \emph{word automata}.
	Word automata are a very useful way to (finitely) represent all the models of an \LTL formula $\psi$.
	As a matter of fact, for every \LTL formula $\psi$, there exists a deterministic parity word automaton $\DAutomaton_{\psi}$ whose language is the set of interpretation on which $\psi$ is true.
	The size of such automaton is double-exponential in the length of $\psi$.
	The following theorem gives precise bounds.
	
	\begin{lemma}[\cite{Pit07}]
		\label{lmm:ltl-to-parity}
		For every \LTL formula $\psi$ over a set $\Var$ of variables, there exists a \emph{deterministic parity automaton} $\DAutomaton_{\psi} = \tuple{2^{\Var}, Q, q_0, \delta, \alpha}$ of size double-exponential w.r.t. $\psi$ and a (single) exponential number of priorities such that $\Lang(\DAutomaton_{\psi}) = \set{\pi \in (2^{\Var})^{\omega}}{\pi \cmodels \psi}$.
	\end{lemma}
	
	A multi-player parity game is a tuple $\Game= \tuple{\Players, (\Actions_{i})_{i \in \Players}, \States, s_{0}, \lambda, \trnFun}$ where
	\begin{inparaenum}[(i)]
		\item $\Players = \{0, \ldots, n\}$ is a set of players;
		\item $\Actions_{i}$ is a set of actions that player $i$ can play;
		\item $\States$ is a set of states with $s_0$ being a designated initial state;
		\item $\lambda: \States \to C$ is a coloring function, assigning a natural number in $C$ to each state of the game;
		\item $\trnFun: \States \times (\Actions_{1} \times \ldots \times \Actions_{n}) \to \States$ is a transition function that prescribe how the game evolves in accordance with the actions taken by the players.
	\end{inparaenum}

	Players identified with an even index are the Even team, whereas the other are the Odd team.
	Objective of the Even team is to generate an infinite play over the set of states whose coloring fulfills the parity condition established by $\lambda$.
	A strategy for Player $i$ of the Even team is a function $\strategy_{i}: \States^{*} \times (\Actions_{1} \times \Actions_{i - 1}) \to \Actions_{i}$, that determines the action to perform in a given instant according to the past history and the current actions of players that perform their choices before $i$.
	
	A tuple of strategies $\tuple{\strategy_{0}, \strategy_{2}, \ldots}$ for the Even team is \emph{winning} if every play that is generated by that, no matter what the Odd team responds, fulfills the parity condition.
	
	Now, consider a \WBQLTL formula of the form $\varphi = \exists X_0 \forall X_1  \ldots \exists X_{n - 1} \forall X_{n} \psi$, with $X_0, X_{n}$ being possibly empty, and $\DAutomaton_{\psi} = \tuple{2^{\Var}, Q, q_0, \delta, \alpha}$ being the \DPW that recognizes the interpretations satisfying $\psi$, with $\alpha = \tuple{F_{0}, F_{1}, \ldots, F_{k}}$.
	Then, consider the multi-player parity game $\Game_{\varphi} = \tuple{\Players, (\Actions_{i})_{i \in \Players}, \States, s_{0}, \lambda, \trnFun}$ where
	\begin{inparaenum}[(i)]
		\item $\Players = \{0, \ldots, n\}$;
		\item $\Actions_{i} = 2^{X_{i}}$ for each $i \in \Players$;
		\item $\States = Q$ with $s_{0} = q_{0}$;
		\item $\lambda: \States \to \SetN$ such that $\lambda(s) = \arg_{j}\{q \in F_{j} \}$;
		\item $\trnFun = \delta$.
	\end{inparaenum}
	The next theorem provides the correctness of this construction.
	
	\begin{theorem}
		\label{thm:wbqltl-satisfiability}
		A \WBQLTL formula $\varphi$ is satisfiable iff there exists a winning strategy for the Even team in the multi-player parity game $\Game_{\varphi}$.
	\end{theorem}

		\begin{proof}
		Observe that every player $i$ is associated to the set of actions $\Actions_{i}$ corresponding to the evaluation of variables in $X_{i}$.
		In addition, every set of existentially quantified variables is associated to a player whose index is even and so playing for the Even team in $\Game_{\varphi}$.
		Also, the ordering of player reflects the order in the quantification prefix $\qntPref$.
		
		In addition to this, note that the strategy tuples for the Even team correspond to the weak-behavioral Skolem functions over $\qntPref$ and so they generate the same set of outcomes over $(2^{X})^{\omega}$.
		
		Since the automaton $\DAutomaton_{\psi}$ accepts all and only those $\omega$-words on which $\psi$ is true, it follows straightforwardly that every weak-behavioral Skolem function $\theta$ over $\qntPref$ is such that $\theta(\pi) \cmodels \psi$ iff $\theta$ is a winning strategy for the Even team in $\Game_{\varphi}$.
		Hence, the \WBQLTL formula $\varphi$ is satisfiable iff $\Game_{\varphi}$ admits a winning strategy for the Even team.
	\end{proof}
		
	Regarding the computational complexity of \WBQLTL, consider that solving a multi-player parity game amounts to decide whether the Even team has a winning strategy in $\Game$.
	A precise complexity result is provided below.
	
	\begin{lemma}{\cite{MMS16}}
		\label{lmm:multi-player-parity}
		The complexity of solving a multi-player parity game $\Game$ is polynomial in the number of states and exponential in the number of colors and players.
	\end{lemma}

	Therefore, we can conclude that the complexity of \WBQLTL satisfiability is as stated below.
	
	\begin{theorem}
		\label{thm:wbqltl-satisfiability-complexity}
		The complexity of \WBQLTL satisfiability is 2EXPTIME-complete
	\end{theorem}

	\begin{proof}
	The procedure described in Theorem~\ref{thm:wbqltl-satisfiability} is 2EXPTIME.
	Indeed, the automata construction of Lemma~\ref{lmm:ltl-to-parity}, produces a game whose set of  states $\States$ is doubly-exponential in $\psi$ and a number of colors $C$ singly exponential in the size of $\psi$.
	Moreover, the number $n$ of players in $\Game_{\varphi}$ is bounded by the length of $\varphi$ itself, as it corresponds to the number of quantifiers in the formula.
	
	Now, from Lemma~\ref{lmm:multi-player-parity}, we obtain that solving $\Game_{\varphi}$ is polynomial in $\States$, and exponential in both $C$ and $n$.
	This amounts to a procedure that is double-exponential in the size of $\varphi$.
	
	Regarding the lower-bound, observe that the formula $\forall X \exists Y \psi$ represents the synthesis problem for the \LTL formula $\psi$ with $X$ and $Y$ being the uncontrollable and controllable variables, which is already 2EXPTIME-Complete~\cite{PR89}.
\end{proof}

\section{Related Work in Formal Methods}

The interaction of second-order quantified variables is of interest in the logic and formal method community.
For instance, Independence-Friendly logic considers dependence atoms as a syntactic extension~\cite{MSS11,GV13}.
Another approach generalizes quantification by means of partially ordered quantifiers~\cite{BG86,KM92} in which existential variables may depend on disjoint sets of universal quantification.

The notion of behavioral has recently drawn the attention of many researchers in the area of logic for strategic reasoning.
Strategy Logic~\cite{MMPV14} (\SL) has been introduced as a formalism for expressing complex strategic and game-theoretic properties.
Strategies in \SL are first class citizens.
Unfortunately, and similarly to \QLTL, quantifications over them sets up a kind of dependence that cannot be realized through actual processes, as they involve future and counter-factual possible computations that are not accessible by reactive programs.
To overcome this, and also mitigate the computational complexities of the main decision problems, the authors introduced a \emph{behavioral} semantics as a way to restrict the dependence among strategies to a realistic one.
They also showed that for a small although significant fragment of \SL, which includes \ATLS, behavioral semantics has the same expressive power of the standard one.
This means that ``\emph{behavioral strategies}'' are able to solve the same set of problems that can be expressed in such fragment.
Further investigations around this notion has been carried out in the community.
In~\cite{GBM18,GBM20}, the authors characterize different notions of behavioral, ruling out future and counter-factual dependence one by one, providing a classification of syntactic fragments for which the behavioral and non-behavioral semantics are equivalent.

\section{Conclusion}
\label{sec:conclusion}

We introduced a behavioral semantic for \QLTL, getting a new logic \emph{Behavioral \QLTL} (\BQLTL). This logic is characterized by the fact that the (second-order) existential
quantification of variables is restricted to depend, at every instant, only on the past interpretations of the variables that are universally quantified upfront in the formula, and not on their entire trace, as it is for classic \QLTL.
This makes such dependence to be a function ready implementable by processes, thus making \BQLTL suitable for capturing advanced forms of planning and synthesis through standard reasoning, as envisioned since in the early days of AI \cite{Green69}.
We studied satisfiability for \BQLTL, providing tight complexity bounds.
For the simplest syntactic fragments, which do not include quantification blocks of the form $\forall X_{i} \exists Y_{i}$, the complexity is the same as \QLTL, given the two semantics are equivalent. For the rest of \BQLTL, where the characteristics of behavioral semantics become apparent, we present an automata-based  technique that is $(n+1)$-EXPTIME, with $n$ being the number of quantification blocks $\forall X_{i} \exists Y_{i}$.
The matching lower-bound comes from a reduction of the corresponding (distributed) synthesis problems.

We also consider a weaker-version of Behavioral \QLTL, denoted \WBQLTL, where the history of quantification is completely visible to every existentially quantified variable, except for the current instant in which only the upfront quantification is available.
We give a technique for satisfiability that is $2$-EXPTIME, regardless of the number of quantifications in the formula.
This is due to the fact that full visibility of variables allows for solving the problem with a simple local reasoning that avoids computationally expensive automata constructions.
Also in this case, the matching lower-bound comes from a reduction of the corresponding synthesis problem, again proving that our technique is optimal.

\begin{paragraph}{Acknowledgments}
	This work is partially supported by ERC Advanced Grant WhiteMech (No. 834228) and the EU ICT-48 2020 project TAILOR (No. 952215). 
\end{paragraph}


\bibliographystyle{plain}
\bibliography{biblio}

%

\end{document}